\documentclass[12pt]{article}
\usepackage[latin9]{inputenc}
\usepackage{amsmath}
\usepackage{amsthm}
\usepackage{amssymb}

\makeatletter

\DeclareTextSymbolDefault{\textquotedbl}{T1}

\theoremstyle{plain}
\newtheorem{thm}{\protect\theoremname}
\theoremstyle{plain}
\newtheorem{lem}[thm]{\protect\lemmaname}
\theoremstyle{definition}
\newtheorem{defn}[thm]{\protect\definitionname}

\usepackage[english]{babel}
\usepackage{fullpage}
\usepackage{amsthm}

\usepackage{bibspacing}
\usepackage{url}

\makeatother

\providecommand{\definitionname}{Definition}
\providecommand{\lemmaname}{Lemma}
\providecommand{\theoremname}{Theorem}

\begin{document}
\title{Communication complexity of pointer chasing via the fixed-set lemma}
\author{Emanuele Viola\thanks{Supported by NSF grant CCF-2430026.}}

\maketitle

The input to the $k$-pointer-chasing function are two arrays $a,b\in[n]^{n}$
of pointers and the goal is to output (say) the first bit of the pointer
reached after following $k$ pointers, starting at $a[0]$. For example,
for $i=0,1,2,3,\ldots$ the output of $i$-pointer chasing is the
first bit of $a[0]$,$b[a[0]]$, $a[b[a[0]]]$, $b[a[b[a[0]]]]$,
$\ldots$. The communication complexity of this fundamental problem
and its variants has a long, rich, and developing history, starting
with \cite{PaS84}, with applications ranging from data structures
to bounded-depth circuits; for background see \cite{KuN97,moti}.
Here I consider deterministic protocols with $k$ rounds and 2 players:
Alice, who sees $b$ but not $a$, and Bob who sees $a$ but not $b$.
(For randomized protocols see \cite{NiW93,DBLP:journals/cpc/Yehudayoff20,DBLP:conf/innovations/MaoYZ25}.)
In a $0$-round protocol Alice computes the answer with no communication
as a function of $b$. In a $1$-round, Alice sends a message $m$
to Bob who then computes the answer as a function of $a$ and $m$;
and so on. Alice always goes first.

\cite{NiW93} prove a $c(n-k\log n)$ lower bound on the communication
complexity. They then write that getting rid of the $-k\log n$ term
``requires a more delicate argument'' which they sketch. The textbook
\cite{RaoY2019} gives a proof.

I give an alternative proof, arguably simpler than the one in \cite{RaoY2019},
of an $n/8$ lower bound. This follows from the next theorem for $A=B=[n]^{n}$
and $F_{A}=F_{B}=\emptyset$. The bound holds for any $k$. In particular,
for $k=n/8$ the bound holds regardless of the number of rounds. I
write $[n]$ for $\{0,1,\ldots,n-1\}$ and for a set $A$ I write
$A$ for its size as well, following notation in \cite{moti}.
\begin{thm}
There is no $k$-round protocol with communication $s$ and sets $A,B\subseteq[n]^{n}$
and $F_{A},F_{B}\subseteq[n]$ such that:

(0) The protocol computes $k$-pointer-chasing on every input in $A\times B$,

(1) The $F_{A}$ pointers in $A$ are fixed, i.e., $\forall i\in F_{A}\exists v\forall a\in A,a[i]=v$,
and the same for $B$,

(2) The unfixed density of $A$, defined as $A/n^{n-F_{A}}$, is $\ge2^{2s-n/4+F_{A}}$,
and the same for $B$,

(3) $A[0]$ is alive, defined as $\mathbb{P}_{a\in A}[a[0]=v]<2/n$
for every $v\in[n]$.
\end{thm}

\begin{proof}
Proceed by induction on $k$, a.k.a.~round elimination. For the base
case $k=0$, we can fix B and hence Alice's output. But since $A[0]$
is alive, the probability that this is correct is $<(2/n)\cdot n/2<1$.

The induction step is by contrapositive. Assuming there is such a
protocol and there are such sets, we construct a $(k-1)$-round protocol
and sets violating the inductive assumption. Suppose Alice sends $t$
bits as her first message. Fix the most likely message, and let $B_{0}\subseteq B$
be the set of $\ge2^{-t}B$ corresponding strings. Next is the key
idea, taken from the proof of the fixed-set Lemma 3.14 from \cite{GrinbergSV-adaptivemajority}.
If there is a pointer $i\in[n]-F_{B}$ which is not alive in $B_{0}$,
fix it to its most likely value, call $B_{1}\subseteq B_{0}$ the
corresponding subset, and let $F_{B_{1}}:=F_{B}\cup\{i\}$. Note that
the unfixed density increases by a factor $2$ since
\[
\frac{B_{1}}{n^{n-F_{B_{1}}}}\ge\frac{2}{n}\frac{B_{0}}{n^{n-F_{B}-1}}=2\frac{B_{0}}{n^{n-F_{B}}}.
\]
Continue fixing until every unfixed pointer is alive, and call $B'$,
$F_{B'}$ the resulting sets. The unfixed density of $B'$ is then
\[
\frac{B'}{n^{n-F_{B'}}}\ge\frac{2^{-t}B}{n^{n-F_{B}}}2^{F_{B'}-F_{B}}\ge2^{2s-n/4+F_{B}-t+F_{B'}-F_{B}}=2^{2(s-t)-n/4+F_{B'}}.
\]

Now note that $F_{B'}\le n/4$ because the density cannot be larger
than $1$. We use this to analyze Alice's side. Because $A[0]$ is
alive, $\mathbb{P}_{a\in A}[a[0]\in F_{B'}]\le2F_{B'}/n\le1/2.$ So
there is an alive pointer $B'[i]$ such that $\mathbb{P}[a[0]=i]\ge1/(2n)$.
Let $A'\subseteq A$ be the corresponding subset with $a[0]=i$, and
let $F_{A'}:=F_{A}\cup\{0\}$. The unfixed density of $A'$ is
\[
\frac{A'}{n^{n-F_{A'}}}\ge\frac{A/2n}{n^{n-F_{A}-1}}=\frac{A/2}{n^{n-F_{A}}}\ge2^{2s-n/4+F_{A}-1}=2^{2s-n/4+F_{A'}-2}\ge2^{2(s-t)-n/4+F_{A'}}
\]
since $t\ge1$. This gives a $(k-1)$-round protocol where Bob goes
first that computes $(k-1)$-pointer-chasing where the first pointer
is $B'[i]$. We can swap players to let Alice go first and permute
pointers to let $A'[0]$ be the first pointer.
\end{proof}
I am grateful to the anonymous referees and Quan Luu for helpful discussions.

\bibliographystyle{alpha}
\bibliography{OmniBib}

\end{document}